\newif\ifPAGELIMIT
\newif\ifCOUNT
\PAGELIMITfalse
\COUNTfalse
\COUNTtrue
\documentclass[conference,letterpaper]{IEEEtran}
\usepackage{amssymb,amsfonts}
\usepackage{amsthm,bm}
\usepackage{pict2e}

\newcommand{\Count}{\ifCOUNT count\else symbol-composition\fi}
\newcommand{\cI}{{\mathcal{I}}}
\newcommand{\bldxi}{{\mathbf{\xi}}}
\newcommand{\bldzero}{{\mathbf{0}}}
\newcommand{\bldone}{{\mathbf{1}}}
\newcommand{\bldf}{{\mathbf{f}}}
\newcommand{\bldp}{{\mathbf{p}}}
\newcommand{\bldrho}{{\mathbf{\rho}}}
\newcommand{\bldr}{{\mathbf{r}}}
\newcommand{\bldu}{{\mathbf{u}}}
\newcommand{\bldw}{{\mathbf{w}}}
\newcommand{\bldx}{{\mathbf{x}}}

\newcommand{\Flist}{{\mathcal{F}}}
\newcommand{\Domain}{{\mathcal{U}}}
\newcommand{\constraint}{{\mathsf{S}}}
\newcommand{\entropy}{{\mathsf{h}}}

\newcommand{\Prob}{{\mathsf{P}}}
\newcommand{\prob}{{\mathsf{Q}}}
\newcommand{\Expected}{{\mathbb{E}}}
\newcommand{\Entropy}{{\mathsf{H}}}
\newcommand{\Information}{{\mathrm{I}}}
\newcommand{\Realfield}{{\mathbb{R}}}
\newcommand{\capacity}{{\mathsf{cap}}}
\newcommand{\OrdCapacity}{\capacity(\constraint_{q;a,b})}
\newcommand{\Capacity}[2]{\capacity(\constraint_{#1},#2)}
\newcommand{\In}{{\mathrm{in}}}
\newcommand{\Out}{{\mathrm{out}}}
\newcommand{\Label}{{\mathsf{L}}}
\newcommand{\half}{{\frac{1}{2}}}
\newcommand{\Low}{L}
\newcommand{\Int}{I}
\newcommand{\High}{H}

\newtheorem{theorem}{Theorem}
\newtheorem{proposition}[theorem]{Proposition}
\newtheorem{lemma}[theorem]{Lemma}

\theoremstyle{remark}

\theoremstyle{definition}

\renewcommand{\mathbf}[1]{{\bm{#1}}}     
\newlength{\figunit}
\begin{document}
\ifCOUNT
\title{The Capacity of Count-Constrained \\ ICI-Free Systems}
\else
\title{The Capacity of ICI-Free Systems Satisfying
Symbol-Composition Constraints}
\fi

\author{%
  \IEEEauthorblockN{Navin Kashyap}
  \IEEEauthorblockA{Dept.\ of Electrical Communication Eng.\\
                    Indian Institute of Science, Bangalore  \\
                    Email: {nkashyap@iisc.ac.in}}
  \and
  \IEEEauthorblockN{Ron M.\ Roth}
  \IEEEauthorblockA{Dept.\ of Computer Science \\
                    Technion, Haifa, Israel \\
                    Email: {ronny@cs.technion.ac.il}}
  \and
  \IEEEauthorblockN{Paul H.\ Siegel}
  \IEEEauthorblockA{Dept.\ of Electrical \& Computer Eng.\ \\
                    Univ.\ of California San Diego, USA \\
                    Email: {psiegel@ucsd.edu}}
}

\maketitle

\begin{abstract}
A Markov chain approach is applied to determine the capacity
of a general class of $q$-ary ICI-free constrained systems that satisfy
an arbitrary \Count\ constraint. 
\end{abstract}

\section{Introduction}
\label{sec:intro}

Let $\Sigma$ be an alphabet of a finite size $q \ge 2$.
A \emph{word} over~$\Sigma$ is any finite string
$\bldw = w_1 w_2 \ldots w_n$ where $w_i \in \Sigma$.
Let~$\Flist$ be a finite set of words over $\Sigma$.
The \emph{(finite-type) constrained system}
$\constraint_\Flist$ consists
of all words $\bldw = w_1 w_2 \ldots w_n$ over ~$\Sigma$
such that~$\Flist$ contains none of their substrings
$w_i w_{i+1}\ldots w_j$, for any $1 \le i \le j \le n$.
We refer to the set~$\Flist$ as the set of
\emph{forbidden words} defining
the constrained system $\constraint_\Flist$.
The constrained system $\constraint_\Flist$
can be presented by a (finite) directed edge-labeled graph~$G$,
with edges labeled with symbols from~$\Sigma$,
such that $\constraint_\Flist$ is the set of all words obtained by
reading off the labels along paths of~$G$. For a proof of this fact,
we refer the reader to~\cite{MRS01}, which provides a comprehensive
introduction to the subject of constrained systems.

Our specific interest is in a general class of
``inter-cell interference free'' (in short, ``ICI-free'')
constrained systems, which we now define.
For prescribed positive integers~$a$, $b$, and~$q$
such that $a+b \le q$,
let $\Sigma$ be an alphabet of size~$q$ which is assumed to be
partitioned into three (disjoint) subsets
$\Low$, $\High$, and~$\Int$, of sizes~$a$, $b$,
and $q-a-b$, respectively.
The elements in $\Low$ (respectively, $\High$) represent
the ``low" (respectively, ``high'') symbols of~$\Sigma$,
while those in $\Int$ are the ``intermediate'' symbols.
The ICI-free constrained system that we consider is
the constrained system\footnote{%
Since only the sizes of $\Sigma$, $\Low$, and~$\High$ will matter,
we identify the constrained system by the sizes of these sets.}
$\constraint_{q;a,b} := \constraint_{\Flist_{q;a,b}}$
defined by the set of forbidden words
$\Flist_{q;a,b} := \{w_1 w_2 w_3: w_1,w_3\in \High, w_2 \in \Low \}$.
A graph $G_{q;a,b}$ presenting the constrained system
$\constraint_{q;a,b}$ is shown in Fig.~\ref{fig:Gqab}.

\newsavebox{\graph}
\sbox{\graph}{%
    \thicklines
    \setlength{\unitlength}{\figunit}
    \multiput(000,000)(060,000){2}{\circle{20}}
    \put(030,040){\circle{20}}
    \put(010,000){\vector(1,0){040}}
    \qbezier(051,-4.7)(030,-14.5)(009,-4.7)
    \put(009.4,-4.7){\vector(-2,1){0}}
    \put(006,008){\vector(3,4){018}}
    \put(036,032){\vector(3,-4){018}}
    \qbezier(039,035)(057,029)(058,010)
    \put(058,010){\vector(1,-4){0}}
    \qbezier(-9.26,004)(-30,015)(-30,000)
    \qbezier(-9.26,-04)(-30,-15)(-30,000)
    \put(-9.26,-04){\vector(2,1){0}}
    \qbezier(069.56,004)(090,015)(090,000)
    \qbezier(069.56,-04)(090,-15)(090,000)
    \put(069.26,-04){\vector(-2,1){0}}
    \qbezier(068.3,006)(100,025)(100,000)
    \qbezier(068.3,-06)(100,-25)(100,000)
    \put(068,-06){\vector(-2,1){0}}
    \put(000,000){\makebox(0,0){$\mathbf{1}$}}
    \put(030,040){\makebox(0,0){$\mathbf{2}$}}
    \put(060,000){\makebox(0,0){$\mathbf{3}$}}
}
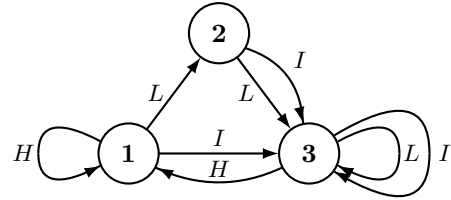
\begin{figure}[hbt]
\begin{center}
\small
\thicklines
\setlength{\unitlength}{\figunit}
\begin{picture}(150,65)(-40,-15)
    \usebox{\graph}
    \put(-35,000){\makebox(0,0){$\High$}}
    \put(030,-05){\makebox(0,0){$\High$}}
    \put(030,005){\makebox(0,0){$\Int$}}
    \put(009,020){\makebox(0,0){$\Low$}}
    \put(039,020){\makebox(0,0){$\Low$}}
    \put(057,031){\makebox(0,0){$\Int$}}
    \put(094,000){\makebox(0,0){$\Low$}}
    \put(105,000){\makebox(0,0){$\Int$}}
\end{picture}
\thinlines
\setlength{\unitlength}{1pt}
\end{center}
	\caption{The graph $G_{q;a,b}$ presenting
        the $q$-ary ICI-free constraint $\constraint_{q;a,b}$.
        Each arrowed line labeled by $X \in \{\Low,\Int,\High\}$
        represents $|X|$ parallel edges labeled by distinct
        symbols from $X$.}
\label{fig:Gqab}
\end{figure}

We additionally impose a \Count\ constraint defined by
a given probability vector
$\bldp = (p_s)_{s \in \Sigma}$
(with nonzero entries that sum to~$1$), which specifies
the frequencies of occurrence of each $s \in \Sigma$ within words
belonging to $\constraint_{q;a,b}$.
To avoid trivialities, we will assume
$\rho_\Low := \sum_{s \in \Low} p_s$ and
$\rho_\High := \sum_{s \in \High} p_s$ to be strictly positive.
The probability
$\rho_\Int := \sum_{s \in \Int} p_s$ is allowed to be~$0$.

For $\varepsilon > 0$,
let $\constraint_{q;a,b}(\bldp,\varepsilon)$ denote
the subset of $\constraint_{q;a,b}$ consisting of all words
$\bldw \in \constraint_{q;a,b}$ in which the number of occurrences of
each symbol $s \in \Sigma$ lies in the interval
$\bigl((p_s - \varepsilon)|\bldw|, (p_s + \varepsilon)|\bldw|\bigr)$,
where $|\bldw|$ denotes the length of~$\bldw$.
The \emph{capacity}
(or the \emph{asymptotic information rate})
of $\constraint_{q;a,b}$ under
the \Count\ constraint specified by~$\bldp$ is defined
as\footnote{All logarithms in this work are to the base $2$.}
\begin{equation}
\label{def:Rqabp}
\Capacity{q;a,b}{\bldp} := \lim_{\varepsilon \rightarrow 0^+}
\limsup_{n \to \infty} \frac{1}{n}
\log |\constraint_{q;a,b}(\bldp,\varepsilon) \cap \Sigma^n|
\; .
\end{equation}
This quantifies, for large~$n$, the exponential rate of growth of
the number of length-$n$ words in $\constraint_{q;a,b}$ in which
the relative frequency of occurrence of each symbol $s \in \Sigma$ is
approximately $p_s$.
Dropping the \Count\ constraint, we also define
the (ordinary) capacity of the constrained system
$\constraint_{q;a,b}$
to be\footnote{By a standard sub-additivity argument,
the limit in this definition exists.}
\begin{equation}
\label{def:Rqab}
\OrdCapacity := \lim_{n \to \infty}
\frac{1}{n} \log |\constraint_{q;a,b} \cap \Sigma^n| \; .
\end{equation}

The quantities $\OrdCapacity$ and $\Capacity{q;a,b}{\bldp}$ were studied
in~\cite{CCK++18}, \cite{QYS13}, \cite{Vu17}, motivated by
proposed coding schemes to mitigate inter-cell interference in
flash memory devices.\footnote{These references used a different
definition of $\capacity(\constraint_{q;a,b},\bldp)$,
which is shown in
\ifPAGELIMIT
Appendix~A in~\cite{KRS}
\else
Appendix~\ref{sec:appendixA}
\fi
to be equivalent to our definition in~(\ref{def:Rqabp}).}
Using standard techniques from the theory of
constrained systems (see e.g., \cite{MRS01})
the (ordinary) capacity $\capacity(\constraint_{q;a,b})$
was shown in~\cite{CCK++18} to be
the largest real root of the cubic polynomial
$x^3 - q x^2 + abx - ab(q-b)$.
The analysis of $\capacity(\constraint_{q;a,b},\bldp)$
in~\cite{CCK++18}
is based on combinatorial arguments, and a Stirling approximation
of the resulting expressions then yields a bivariate function
which needs to be maximized (numerically) in order to 
obtain the values of $\capacity(\constraint_{q;a,b},\bldp)$.

In this work, we make use of a result from~\cite{MR92}
to formulate the problem of determining
the capacity $\Capacity{q;a,b}{\bldp}$
as an optimization problem over Markov chains defined on
the graph $G_{q;a,b}$ shown in Fig.~\ref{fig:Gqab}.
By shifting to the dual optimization problem, we then derive
an analytical solution to this optimization problem,
which results in
an exact expression for $\Capacity{q;a,b}{\bldp}$ given in
Theorems~\ref{thm:R3p} and~\ref{thm:Rqp} in Section~\ref{sec:Rqabp}.
While our analysis is tailored to
\Count-constrained ICI-free systems,
some of the tools that we use may be applicable to other
constrained systems as well (see~\cite{EMS2016}).

\section{Markov Chains and Optimization}
\label{sec:mc}

Let~$G = (V,E)$ be a directed graph with vertex set $V$ and
(directed) edge set $E$. For a vertex $v \in V$,
we let $E_\In(v)$ and $E_\Out(v)$ denote the set of incoming and
outgoing edges, respectively, incident with~$v$.

A stationary Markov chain on~$G$ is a probability distribution
$P = \bigl(P(e)\bigr)_{e \in E}$ on $E$, with the property that
for each $v \in V$, the sum of the probabilities on
the incoming edges of~$v$ is equal to that on the outgoing edges of~$v$:
\begin{equation}
\label{eq:Kirchhoff}
\sum_{e \in E_\In(v)} P(e)
= \sum_{e \in E_\Out(v)} P(e) \; .
\end{equation}
The induced stationary distribution on the vertex set $V$ is given by 
$\pi(v) = \sum_{e \in E_\Out(v)} P(e)$, for all $v \in V$. The set of
all stationary Markov chains on~$G$ is denoted by $\Delta(G)$.

The \emph{entropy rate} of a stationary Markov chain $P$ on $G$
is defined as
\[
\Entropy(P) := - \sum_{e \in E}
P(e) \log P(e) - \biggl(-\sum_{v \in V} \pi(v) \log \pi(v)\biggr) \; .
\]
Since
$\Entropy(P) =
- \sum_{v \in V} \sum_{e \in E_\Out(v)} P(e) \log (P(e)/\pi(v))$,
the convexity properties of relative entropy imply that
$P \mapsto \Entropy(P)$ is a concave function.

Given a Markov chain $P \in \Delta(G)$,
along with a vector of real-valued functions
$\bldf = (f_1 \; f_2 \; \ldots \; f_t) : E \rightarrow \Realfield^t$,
we denote by $\Expected_P(\bldf)$ the expected value of~$\bldf$
with respect to~$P$:
\[
\Expected_P(\bldf) =
\sum_{e \in E} P(e) \bldf(e) \; .
\]
We will only need the following special case of the function~$\bldf$.
Let $\Label : E \to \Sigma$ be a labeling of the edges of
the graph~$G$ with symbols from~$\Sigma$.
For a subset $W$ of~$\Sigma$ of size~$t$,
we define the vector indicator function
$\cI_W : E \rightarrow \Realfield^t$
by $\cI_W = (\cI_s)_{s \in W}$,
where $\cI_s : E \rightarrow \Realfield$ is the indicator function
for a symbol $s \in \Sigma$:
\[
\cI_s (e) = \left\{ \begin{array}{lcl}
                           1 & & \textrm{if $\Label (e) = s$} \\
                           0 & & \textrm{otherwise}
                         \end{array}
\right. \; .
\]
Then, $\Expected_P(\cI_W)$ is a vector in $\Realfield^t$
whose entry that is indexed by $s \in W$ is the probability that
an edge chosen according to
the distribution~$P$ is labeled with the symbol~$s$.

These definitions allow us to state the following result, which
expresses $\Capacity{q;a,b}{\bldp}$ as the solution to
a convex optimization problem.

\begin{proposition}
We have
\label{prop:Rqabp}
\[
\Capacity{q;a,b}{\bldp} =
\sup_{P \in \Delta(G_{q;a,b}) \,: \atop \Expected_P(\cI_W) = \bldp'}
\Entropy(P) \; ,
\]
for any $W \subset \Sigma$ of size $q-1$
and $\bldp' = (p_s)_{s \in W}$.
\end{proposition}

\begin{proof}
As a consequence of~\cite[Lemma~2]{MR92},
for any $\varepsilon > 0$,
$\limsup_{n \to \infty}
(1/n) \log |\constraint_{q;a,b}(\bldp,\varepsilon) \cap \Sigma^n|$
is equal to $\sup \Entropy(P)$, the supremum being over stationary
Markov chains $P \in \Delta(G_{q;a,b})$ such that
$\Expected_P(\cI_\Sigma)
\in (\bldp-\varepsilon \cdot \bldone,\bldp+\varepsilon \cdot \bldone)$
(where~$\bldone$ denotes the all-one vector in $\Realfield^q$).
We claim that as $\varepsilon \to 0$,
these suprema converge to $\sup \Entropy(P)$, 
the supremum now being over stationary
Markov chains $P \in \Delta(G_{q;a,b})$ such that
$\Expected_P(\cI_\Sigma) = \bldp$. With this, we would have 
\begin{equation}
\label{eq:Rqabp}
\Capacity{q;a,b}{\bldp} =
\sup_{P \in \Delta(G_{q;a,b}) \, : \atop \Expected_P(\cI_\Sigma)= \bldp}
\Entropy(P) \; .
\end{equation}
The constraint $\Expected_P(\cI_\Sigma) = \bldp$ in the supremum on
the right-hand side (RHS) above can be replaced by
$\Expected_P(\cI_W) = (p_s)_{s \in W}$,
since the latter implies $\Expected_P(\cI_{\{ s \}}) = p_s$
for the remaining symbol $s \in \Sigma \setminus W$.
This would prove the proposition.

We now prove the claim above. To this end, for $\varepsilon >  0$,
define $\Delta_{\bldp,\varepsilon}$ to be the set of all stationary
Markov chains $P \in \Delta(G_{q;a,b})$ such that
$\Expected_P(\cI_\Sigma)
\in (\bldp-\varepsilon \cdot \bldone,\bldp+\varepsilon \cdot \bldone)$. 
Its closure $\overline{\Delta_{\bldp,\varepsilon}}$ is the set of all
$P \in \Delta(G_{q;a,b})$ such that
$\Expected_P(\cI_\Sigma)
\in [\bldp-\varepsilon \cdot \bldone,\bldp+\varepsilon \cdot \bldone]$. 
By continuity of the mapping $P \mapsto \Entropy(P)$, we have
\[
\sup_{P \in \Delta_{\bldp,\varepsilon}} \Entropy(P)
= \sup_{P \in \overline{\Delta_{\bldp,\varepsilon}}} \Entropy(P),
\]
and the latter supremum is in fact a maximum.
Finally, let $\Delta_{\bldp,0}$ denote
the set of all $P \in \Delta(G_{q;a,b})$ such that
$\Expected_P(\cI_\Sigma) = \bldp$. We wish to show that
\begin{equation}
\label{eq:sup}
\lim_{\varepsilon \rightarrow 0^+}
\sup_{P \in \overline{\Delta_{\bldp,\varepsilon}}} \Entropy(P) = 
\sup_{P \in \Delta_{\bldp,0}} \Entropy(P).
\end{equation}
The limit on the left-hand side (LHS)
of~(\ref{eq:sup}) exists since
$\sup_{P \in \overline{\Delta_{\bldp,\varepsilon}}} \Entropy(P)$
is a monotone function of~$\varepsilon$.

Since $\Delta_{\bldp,0} \subseteq \overline{\Delta_{\bldp,\varepsilon}}$
for all $\varepsilon > 0$, the RHS above cannot exceed the LHS.
To prove the reverse inequality, suppose that $P_{\varepsilon}$
achieves the supremum over
$P \in \overline{\Delta_{\bldp,\varepsilon}}$.
Passing to a subsequence if necessary, $P_{\varepsilon}$
converges (as $\varepsilon \to 0^+$)
to some $P_0 \in \Delta(G_{q;a,b})$. From the fact that
$\Expected_P(\cI_\Sigma)$ is continuous in $P$,
it follows that $P_0 \in \Delta_{\bldp,0}$.
Hence, again via the continuity of
the mapping $P \mapsto \Entropy(P)$, we obtain
\[
\lim_{\varepsilon \rightarrow 0^+}
\sup_{P \in \overline{\Delta_{\bldp,\varepsilon}}} \Entropy(P) = 
\lim_{\varepsilon \rightarrow 0^+} \Entropy(P_{\varepsilon}) = 
\Entropy(P_0) \le \sup_{P \in \Delta_{\bldp,0}} \Entropy(P),
\]
which proves our claim.
\end{proof}

Thus, computation of the quantity $\Capacity{q;a,b}{\bldp}$ requires
the solution of a constrained optimization problem in which
the objective function $P \mapsto \Entropy(P)$ is concave,
and the constraints are linear.
The theory of convex duality based upon Lagrange multipliers
provides a method to translate the problem into an unconstrained
optimization with a convex objective function~\cite{MR92}.

In order to reformulate the problem, we need to introduce
a vector-valued matrix function that
generalizes the adjacency matrix of a directed graph $G = (V,E)$.
For a function $\bldf: E \rightarrow \Realfield^t$
and $\bldxi \in \Realfield^t$, let
$A_{G;\bldf} (\bldxi)$ be the matrix defined by
\[
\Bigl( \, A_{G;\bldf} (\bldxi) \, \Bigr)_{u,v} =
\sum_{e \in E_\Out(u) \cap E_\In(v)}
2^{-\bldxi \cdot \bldf(e)} \; .
\]
We remark that for any
function~$\bldf$, the matrix $A_{G;\bldf} (\bldzero)$ is precisely
the adjacency matrix of $G$. Moreover, for any choice of
$\bldxi \in \Realfield^t$,
the matrix $A_{G;\bldf} (\bldxi)$ is (entry-wise) non-negative,
so that it has a unique largest positive eigenvalue, called
the Perron eigenvalue,
which we denote by $\lambda(A_{G;\bldf} (\bldxi))$.

The following lemma is the main tool in translating the constrained
optimization problem to a more tractable form.
It is a consequence of standard results in the theory of convex duality.

\begin{lemma}
\label{lem:dual}
Let~$G$ and~$\bldf$ be as above. Then, for any $\bldr \in \Realfield^t$,
\[
\sup_
{P \in \Delta(G) \,: \atop \Expected_P(\bldf) = \bldr}
\Entropy(P) =
\inf_{\bldxi \in \Realfield^t}
\left\{ \, \bldxi \cdot \bldr +
\log \lambda(A_{G;\bldf}(\bldxi)) \, \right\} \; .
\]
\end{lemma}

Note that since $P \mapsto \Entropy(P)$ is a concave function,
by convex duality, the objective function on the RHS
of the lemma is a convex function of $\bldxi$.
Moreover, it is a differentiable function of~$\xi$
whenever the graph~$G$ is strongly-connected
(as is the case when $G = G_{q;a,b}$):
the matrix $A_{G;\bldf} (\bldxi)$ is then irreducible
for all $\bldxi \in \Realfield^t$, so that its
Perron eigenvalue is simple, and hence differentiable as
a function of~$\bldxi$. Consequently, the objective function can be
minimized by identifying the point at which its gradient
with respect to $\bldxi$ vanishes.

We illustrate the use of
Proposition~\ref{prop:Rqabp} and Lemma~\ref{lem:dual}
to determine $\Capacity{q;a,b}{\bldp}$ in the case of $q = 3$
in Section~\ref{sec:q=3}. We will later show in Section~\ref{sec:q>=3}
that the general $q \ge 3$ case can be reduced to $q = 3$.

\section{Computation of $\Capacity{q;a,b}{\bldp}$}
\label{sec:Rqabp}

The simplest case is that of $q = 2$, i.e.,
the $\constraint_{2;1,1}$ constrained system. This is
the ``no-$101$'' constrained system, which forbids the occurrence of
the string $101$. The value of $\Capacity{2;1,1}{(1{-}p,p)}$,
for $p \in (0,1)$, can be computed via
Proposition~\ref{prop:Rqabp} and Lemma~\ref{lem:dual},
using an analysis similar to (but simpler than) that in
Section~\ref{sec:q=3}. However, we do not provide
the details of this analysis, as it is not difficult to convince
oneself that
$\Capacity{2;1,1}{(1{-}p,p)} =
\Capacity{3;1,1}{(1{-}p,0,p)}$.
Thus, we start with the $q = 3$ case.

\subsection{The Case $q=3$ and $a=b=1$}
\label{sec:q=3}

The key to our analysis of the capacity $\Capacity{q;a,b}{\bldp}$
is the case $(q;a,b) = (3;1,1)$. As noted above, this case subsumes
the case $(q;a,b) = (2;1,1)$. Moreover, as we will show
in the next subsection, the computation of $\Capacity{q;a,b}{\bldp}$
for any $q \ge 3$, $a \ge 1$, and $b \ge 1$ can be reduced
to the problem of computing
$\Capacity{3;1,1}{\bldrho}$, where the entries of~$\bldrho$
are $\rho_X = \sum_{s \in X} p_s$,
for $X \in \{ \Low, \Int, \High \}$.

So, consider a ternary alphabet $\Sigma$ partitioned
into singleton subsets $\Low$, $\Int$, and~$\High$.
By abuse of notation,
we will assume that $\Low$, $\Int$, and $\High$ are the actual
elements of the alphabet~$\Sigma$.
The graph presentation of $\constraint_{3;1,1}$
is given by Fig.~\ref{fig:Gqab},
regarding each arrowed line in the figure as a single edge.

Let the \Count\ constraint vector be
$\bldrho = (\rho_\Low,\rho_\Int,\rho_\High)$,
with $\rho_\Low,\rho_\High \in (0,1)$ and $\rho_\Int \in [0,1)$.
  From Proposition~\ref{prop:Rqabp} and Lemma~\ref{lem:dual}
(applied with $\bldf = (\cI_\Int,\cI_\High)$), we obtain
\begin{eqnarray}
\lefteqn{ \Capacity{3;1,1}{\bldrho}
} \makebox[0ex]{}\nonumber \\
& = & 
\!\!\!\!\!\!\!\!\!\!
\inf_{(\xi_\Int,\xi_\High) \in \Realfield^2}
\left\{\rho_\Int \xi_\Int + \rho_\High \xi_\High 
+ \log \lambda(A_{G;(\cI_\Int,\cI_\High)}(\xi_\Int,\xi_\High))\right\}
\; , \nonumber \\
\label{eq:R3p:1}
\end{eqnarray}
where
\begin{equation}
A_{G;(\cI_\Int,\cI_\High)}(\xi_\Int,\xi_\High)  = \left(
\begin{array}{ccc}
2^{-\xi_\High} & 1 & 2^{-\xi_\Int} \\
0 & 0 & 1 + 2^{-\xi_\Int} \\
2^{-\xi_\High} & 0 & 1 + 2^{-\xi_\Int}
\end{array}
\right) \; .
\label{eq:Axi}
\end{equation}
As noted after Lemma~\ref{lem:dual}, the objective function on
the RHS of~(\ref{eq:R3p:1}) can be minimized by identifying
the point at which its gradient with respect to
$(\xi_\Int,\xi_\High)$ equals~$0$.

The case $\rho_\Int = 0$ needs a little extra care, as in
this case the infimum in~(\ref{eq:R3p:1}) is achieved by
letting $\xi_\Int \to \infty$. This follows from the fact that
for any fixed $\xi_\High$, the Perron eigenvalue
$\lambda(A_{G;(\cI_\Int,\cI_\High)}(\xi_\Int,\xi_\High))$
is strictly decreasing in $\xi_\Int$ (see Problem~3.12 in \cite{MRS01}).
Thus, the RHS of~(\ref{eq:R3p:1}) reduces to
the single-variable optimization problem
$\inf_{\xi_\High} \{\rho_\High \xi_\High
+ \log \lambda(A_{G;\cI_\High}(\xi_\High))\}$,
where $A_{G;\cI_\High}(\xi_\High)$ is the matrix obtained
by setting $2^{-\xi_\Int} = 0$ in~(\ref{eq:Axi}).

We first assume that $\rho_\Int > 0$
(describing later the minor modifications to be made to handle
the case $\rho_\Int = 0$).
We make the change of variables
$y = 2^{-\xi_\Int}$ and $z = 2^{-\xi_\High}$ to get
\begin{equation}
\label{eq:R3p:2}
\Capacity{3;1,1}{\bldrho} = \log \left( \inf_{y,z \in (0,\infty)^2}
\frac{\lambda(A(y,z))}{y^{\rho_\Int}z^{\rho_\High}} \right) \; ,
\end{equation}
where $\lambda(A(y,z))$ is the Perron eigenvalue of the matrix
\[
A(y,z) := \left(
\begin{array}{ccc}
z & 1 & y \\
0 & 0 & 1+y \\
z & 0 & 1+y
\end{array}
\right) \; .
\]
It is easily checked that the determinant of the Jacobian of
the transformation $(\xi_\Int,\xi_\High) \mapsto (y,z)$ is nonzero
for all $(\xi_\Int,\xi_\High) \in \Realfield^2$.
It follows from this that
for any $(\xi_\Int,\xi_\High) \in \Realfield^2$,
the gradient of the objective function in~(\ref{eq:R3p:1}) is~$0$
at $(\xi_\Int,\xi_\High)$ if and only if the gradient of
the objective function in~(\ref{eq:R3p:2}) is~$0$
at $(y,z) = (2^{-\xi_\Int},2^{-\xi_\High})$.
Thus, the minimization in~(\ref{eq:R3p:2}) can be carried out
by identifying the positive values of~$y, z$ at which the gradient of
$\lambda(A(y,z))/(y^{\rho_\Int}z^{\rho_\High})$ vanishes.

To do this, we make another convenient change of variables:
$(y,z) \mapsto (y,\lambda)$ with $\lambda = \lambda(A(y,z))$.
This mapping is invertible: since $A(y,z)$ is irreducible for $y,z > 0$,
it follows from Problem~3.12 in~\cite{MRS01}
that $\lambda(A(y,z))$ is strictly increasing in~$z$
for every fixed~$y > 0$. Also, for each fixed $y > 0$,
the mapping $z \mapsto \lambda(A(y,z))$ is
a continuous function from $(0,\infty)$ onto $(y+1,\infty)$
(as it is easy to see that $\lambda(A(y,0)) = y+1$).
For every fixed~$y > 0$, the inverse mapping
$(y,\lambda) \mapsto (y,z)$ is determined by setting
the characteristic polynomial of $A(y,z)$ equal to $0$, and is given by
\begin{equation}
\label{eq:z}
z = z(y,\lambda) =
\frac{\lambda^2(\lambda-y-1)}{\lambda^2 - \lambda + y+1} \; .
\end{equation}
It can be verified by direct computation\footnote{Also see
\ifPAGELIMIT
Appendix~B in~\cite{KRS}
\else
Appendix~\ref{sec:appendixB}
\fi
for an argument using Perron--Frobenius theory.}
that $\partial z / \partial \lambda > 0$ whenever $\lambda > y+1 > 1$,
and hence, the Jacobian determinant of
the transformation $(y,z) \mapsto (y,\lambda)$ is nonzero
for all $y,z > 0$. From this, arguing as above for
the mapping $(\xi_\Int,\xi_\High) \mapsto (y,z)$, we obtain
via~(\ref{eq:R3p:2}) and~(\ref{eq:z}) that
\begin{equation}
\label{eq:R3p:3}
\Capacity{3;1,1}{\bldrho}
= \log \left(\inf_{(y,\lambda) \in \Domain} g(y,\lambda) \right),
\end{equation}
where
\[
g(y,\lambda) = \frac{\lambda}{y^{\rho_\Int}}
\left(
\frac{\lambda^2 - \lambda + y+1}{\lambda^2(\lambda-y-1)}
\right)^{\rho_\High}
\]
and
\[
\Domain = \{ (y,\lambda) \in \Realfield^2: \lambda > y + 1 > 1 \} \; .
\]
Moreover, the infimum in~(\ref{eq:R3p:3}) is obtained
at any point $(y,\lambda) \in \Domain$ where the partial derivatives of
$g(y,\lambda)$ vanish.

Turning now to the case $\rho_\Int = 0$,
it can be handled by setting $y = 0$ in
the discussion above, assuming the convention that $0^0 = 1$.
Thus, the RHS of~(\ref{eq:R3p:2}) reduces to
$\log \left(\inf_{z > 0} \lambda(A(0,z))/z^{\rho_\High}\right)$,
and the RHS of~(\ref{eq:R3p:3}) becomes
$\log \left(\inf_{\lambda > 1} g_0(\lambda)\right)$,
where $g_0(\lambda) := g(0,\lambda)$.
This latter infimum is achieved at any point $\lambda > 1$
where the derivative $g_0'(\lambda)$ equals~$0$.

In
\ifPAGELIMIT
Appendix~C in the full version of this paper~\cite{KRS},
\else
Appendix~\ref{sec:appendixC},
\fi
we compute the partial derivatives
$\partial g / \partial y$ and $\partial g / \partial \lambda$,
each being a cubic multinomial in~$y$ and~$\lambda$.
We then find explicitly their common root,
thereby yielding the following result.

\begin{theorem}
\label{thm:R3p}
For
$\bldrho =
(\rho_\Low,\rho_\Int,\rho_\High) \in (0,1) \times [0,1) \times (0,1)$:
\[
\Capacity{3;1,1}{\bldrho}
= \log \left[\frac{\lambda}{y^{\rho_\Int}}
\left(
\frac{\lambda^2 - \lambda + y+1}{\lambda^2(\lambda-y-1)}
\right)^{\rho_\High}\right]
,
\]
where $(y,\lambda)$ is given as follows.
\begin{itemize}
\item
If $\rho_\Low = \half$ and $\rho_\Int = 0$,
then $y = 0$ and $\lambda = 2$.
\item
If $\rho_\Low = \half$ and $\rho_\Int > 0$,
then
\[
y =
- 1 - 2\tau + 2 \sqrt{1 + \tau + \tau^2}
\quad \mbox{and} \quad \lambda = 2 \; ,
\]
where $\tau := \rho_\High/\rho_\Int$.
\item
If $\rho_\Low \ne \half$, then
\[
y = \frac{\rho_\Int(\lambda-2)}{1-2\rho_\Low}
\]
and~$\lambda$ is a root of the cubic polynomial
\begin{eqnarray*}
\lefteqn{
Z(x) := (1-2\rho_\Low) \rho_\Low x^3
+ \bigl( (\rho_\Low-\rho_\High)^2 - (1-2\rho_\Low) \bigr) x^2
} \makebox[9ex]{} \\
&& {} - 2(\rho_\Low -\rho_\High)(1-2\rho_\High) x + (1 - 2\rho_\High)^2
\end{eqnarray*}
chosen as follows: if $\rho_\Low < \half$, then~$\lambda$ is
the largest real (positive) root of $Z(x)$; and if $\rho_\Low > \half$,
then $\lambda$ is the smallest real (positive) root of $Z(x)$.
\end{itemize}
\end{theorem}

It is worth noting that for $\bldrho = (\half,0,\half)$
we obtain $\Capacity{3;1,1}{(\half,0,\half)} = \half \log 3$.
Thus, $\Capacity{2;1,1}{(\half,\half)} = \half \log3$,
which agrees with the rate derived
(using two different approaches) in~\cite{QYS13}.

\subsection{The General Case of $q \ge 3$}
\label{sec:q>=3}

Consider now a constrained system $\constraint_{q;a,b}$
over a $q$-ary alphabet $\Sigma$
for some $q \ge 3$, where~$\Sigma$ is partitioned into
the subsets $\Low$, $\High$, and~$\Int$ of sizes
$a \ge 1$, $b \ge 1$, and $q - a - b \ge 0$, respectively.
Let $\bldp = (p_s)_{s \in \Sigma}$ be
a given \Count\ constraint vector,
and define $\rho_X = \sum_{s \in X} p_s$
for $X \in \{\Low,\Int,\High\}$.
If $\Int = \emptyset$, we set $\rho_\Int = 0$.
The probabilities $\rho_\Low$ and $\rho_\High$ are assumed
to be strictly positive.

The aim of this subsection is to prove the result stated next.
The statement requires the following standard definition:
the \emph{entropy} of a probability vector
$\bldu = (u_i)_i$ is defined
as $\entropy(\bldu) = -\sum_i u_i \log u_i$.

\begin{theorem}
\label{thm:Rqp}
For $\constraint_{q;a,b}$ and $\bldp$ as above:
\[
\Capacity{q;a,b}{\bldp} =
\Capacity{3;1,1}{\bldrho} + \entropy(\bldp)
- \entropy(\bldrho) \; ,
\]
where the entries of $\bldrho$ are $\rho_X = \sum_{s \in X} p_s$,
for $X \in \{ \Low, \Int, \High \}$.
\end{theorem}

Thus, the computation of $\Capacity{q;a,b}{\bldp}$ reduces to
the problem of computing
$\Capacity{3;1,1}{\bldrho}$,
which was solved explicitly in Theorem~\ref{thm:R3p}.
The rest of this subsection is devoted to
a proof of Theorem~\ref{thm:Rqp}.

Let $\Prob = {(\Prob(u,s))}_{u,s}$ be a stationary Markov chain
on the labeled graph $G_{q;a,b}$ on
the vertex set $V = \{1,2,3\}$ in Fig.~\ref{fig:Gqab},
where $\Prob(u,s)$ is the probability of the edge labeled~$s$
leaving vertex~$u$
(and $\Prob(u,s) \equiv 0$ if there is no such edge).
Note that for any $s \in \Sigma$,
we have $\Expected_{\Prob}(\cI_s) = \sum_{u \in V} \Prob(u,s)$.
Thus, the constraint $\Expected_{\Prob}(\cI_\Sigma) = \bldp$ on
the RHS of~(\ref{eq:Rqabp}) is equivalently expressed as
\begin{equation}
\label{eq:constraint}
\sum_{u \in V} \Prob(u,s) = p_s \; , \quad
\textrm{for all} \; s \in \Sigma \; .
\end{equation}

Now, for $u \in V$ and $X \in \{ \Low,\Int,\High \}$, define
\[
\prob(u,X) = \sum_{s \in X} \Prob(u,s) \; .
\]
Note that $\prob = {(\prob(u,X))}_{u,X}$ is
a stationary Markov chain on the graph in Fig.~\ref{fig:Gqab},
where each arrowed line in the figure is regarded as a single edge
(this graph is $G_{3;1,1}$).
Moreover, we have for $X \in \{\Low,\Int,\High\}$,
\[
\Expected_{\prob}(\cI_X) = \sum_{u \in V} \prob(u,X)
= \sum_{s \in X} \sum_{u \in V}  \Prob(u,s) \; .
\]
Thus, if we impose the constraint~(\ref{eq:constraint}) on
the Markov chain $\Prob$, we obtain
\[
\Expected_{\prob}(\cI_X)
= \sum_{s \in X} p_s = \rho_X \; ,
\textrm{for all} \; X \in \{\Low,\Int,\High\} \; .
\]
In other words, the constraint $\Expected_{\Prob}(\cI_\Sigma) = \bldp$
on the Markov chain $\Prob$ induces the constraint
$\Expected_{\prob}(\cI_{\{\Low,\Int,\High\}}) = \bldrho$
on the Markov chain~$\prob$.
Finally, observe that $\Prob$ and~$\prob$
induce the same stationary distribution on $V$:
\begin{eqnarray*}
\pi_{\Prob}(u) = \sum_{s \in \Sigma} \Prob(u,s)
& = & \sum_X \sum_{s \in X} \Prob(u,s) \\
& = &
\sum_X \prob(u,X) = \pi_{\prob}(u) \; .
\end{eqnarray*}

The following lemma is the key to proving Theorem~\ref{thm:Rqp}.

\begin{lemma}
\label{lem:entropy}
For a Markov chain $\Prob \in \Delta(G_{q;a,b})$ with
$\Expected_{\Prob}(\cI_\Sigma) = \bldp$,
and $\prob \in \Delta(G_{3;1,1})$ as above, we have
\begin{equation}
\label{eq:entropy}
\Entropy(\Prob) \le
\Entropy(\prob)
+ \entropy(\bldp) - \entropy(\bldrho) \; ,
\end{equation}
with equality holding if and only if
$\Prob(u,s) = (p_s/\rho_X) \prob(u,X)$
for every $s \in X$
(where $\Prob(u,s) = 0$ when $p_s = \rho_X = 0$).
\end{lemma}

\begin{proof}
Let $(U,S)$ be a pair of random variables taking values
$(u,s) \in V \times \Sigma$ with probability
$\Prob(u,s)$. Let
$\varphi : \Sigma \to \{\Low,\Int,\High\}$ be the function that maps~$s$
to~$X$ if $s \in X$.
Now, $U$---$S$---$\varphi(S)$ is a Markov chain, so that by
the data processing inequality,
$\Information(U;S) \ge \Information(U;\varphi(S))$.
It is easily verified that
$\Information(U;S) = \entropy(\bldp) - \Entropy(\Prob)$ and
$\Information(U;\varphi(S))=\entropy(\bldrho)-\Entropy(\prob)$.
Thus,
\[
\entropy(\bldp) - \Entropy(\Prob)
\ge \entropy(\bldrho) - \Entropy(\prob) \; ,
\]
re-arranging which we obtain~(\ref{eq:entropy}).

Equality holds in the data processing inequality above
if and only if
$U$---$\varphi(S)$---$S$ is also a Markov chain, i.e.,
$U$ and~$S$ are conditionally independent given $\varphi(S)$.
Now, check that
$\Pr \{ U=u,S=s \mid \varphi(S) = X \}$
equals
$\Prob(u,s)/\rho_X$
if $s \in X$, and equals~$0$ otherwise.
Hence,
$\Pr \{ U=u,S=s \mid \varphi(S) = X \} = \sum_{s \in X}
\Prob(u,s)/\rho_X = \prob(u,X)/\rho_X$.
Finally, $\Pr \{ S=s \mid \varphi(S) = X \}$
equals $p_s/\rho_X$ if $s \in X$,
and equals~$0$ otherwise. Thus, the required conditional independence
holds if and only if
$\Prob(u,s) = \prob(u,X) \, (p_s/\rho_X)$
for all $s \in X$.
\end{proof}

We can now complete the proof of Theorem~\ref{thm:Rqp}.
Taking the supremum over $\Prob$ in~(\ref{eq:entropy}),
we obtain (by virtue of~(\ref{eq:Rqabp})) that
\begin{equation}
\label{eq:Rineq}
\Capacity{q;a,b}{\bldp}
\le \Capacity{3;1,1}{\bldrho}
+ \entropy(\bldp) - \entropy(\bldrho) \;.
\end{equation}
We now argue that this is in fact an equality.
Consider a $\prob^* = {\bigl(\prob^*(u,X)\bigr)}_{u,X}$
that achieves
$\Capacity{3;1,1}{\bldrho} = \sup \Entropy(\prob)$,
the supremum being over Markov chains
$\prob \in \Delta(G_{3;1,1})$ such that
$\Expected_{\prob}(\cI_{\{\Low,\Int,\High\}}) = \bldrho$.
Such a $\prob^*$ exists as $\prob \mapsto \Entropy(\prob)$ is
a continuous function being maximized over a compact set.
Recall that any outgoing edge from~$u$ labeled by~$X$ in
$G_{3;1,1}$ is replaced in $G_{q;a,b}$ by $|X|$
parallel edges labeled by the distinct symbols $s \in X$.
For each such edge $(u,s)$, set
$\Prob(u,s) = (p_s/\rho_X) \prob^*(u,X)$.
The resulting Markov chain $\Prob \in \Delta(G_{q;a,b})$
satisfies the conditions for equality in~(\ref{eq:entropy}), from
which it follows that equality holds in~(\ref{eq:Rineq}).

\section{Discussion}
\label{sec:discuss}

Our computation of
$\Capacity{q;a,b}{\bldp}$ consists of the following steps.
\begin{enumerate}
\item
\label{item:1}
Applying Theorem~\ref{thm:Rqp} to reduce the problem to that
of computing $\Capacity{3;1,1}{\bldrho}$.
\item
\label{item:2}
Expressing the computation of $\Capacity{3;1,1}{\bldrho}$
as the bivariate minimization problem~(\ref{eq:R3p:2})
in the variables~$(y,z)$.
\item
\label{item:3}
Eliminating the implicit expression $\lambda(A(y,z))$
in~(\ref{eq:R3p:2}) through a change of variables,
resulting in the bivariate minimization problem~(\ref{eq:R3p:3})
in the variables~$(y,\lambda)$.
\item
\label{item:4}
Taking partial derivatives with respect to~$y$ and~$\lambda$,
resulting in two cubic bivariate polynomials in~$y$~and~$\lambda$.
\item
\label{item:5}
Finding the common root of these polynomials.
\end{enumerate}

While Step~\ref{item:1} is specific to
the constrained system~$\constraint_{q;a,b}$, the other steps
might be applicable to other \Count-constrained systems
(albeit with varying degrees of difficulty).
For any constrained system~$\constraint$ over an alphabet~$\Sigma$,
the number of variables in Step~\ref{item:2} will be $|\Sigma|-1$.
As for Step~\ref{item:3}, the explicit
rational expression~(\ref{eq:z}) for~$z = z(y,\lambda)$
is attributed to the fact that the coefficients of
the characteristic polynomial of $A(y,z)$
are linear terms in~$z$. In general, this happens
whenever there is a symbol $s \in \Sigma$ that has a ``home state''
in the graph presentation of~$\constraint$, namely,
all edges labeled by~$s$ lead to the same vertex.

We mention that one could also compute
$\Capacity{q;a,b}{\bldp}$ based
on Proposition~\ref{prop:Rqabp} directly.
Referring to the case $\Capacity{3;1,1}{\bldrho}$
and using the notation towards the end of Section~\ref{sec:q>=3},
such a computation would entail finding the nine edge probabilities of
a Markov chain $\prob = {\bigl(\prob(u,X)\bigr)}_{u,X}$
(where $u \in V = \{ 1, 2, 3 \}$ and
$X \in \Sigma = \{ \Low, \Int, \High \}$)
that maximizes $\Entropy(\prob)$,
subject to the following six linear constraints:
\begin{itemize}
\item
$\prob(2,\High) = 0$,
\item
the constraints~(\ref{eq:Kirchhoff}) for any two vertices in~$V$
(the third is dependent on these two), and---
\item
the three constraints obtained from
$\Expected(\cI_\Sigma) = \bldrho$
(these constraints imply that $\sum_{u,X} \prob(u,X) = 1$).
\end{itemize}
We would then end up with three linearly independent variables
to optimize over.

\section*{Acknowledgment}

N.\ Kashyap and P.~H.\ Siegel would like to acknowledge
the support of the Indo--US Science \& Technology Forum (IUSSTF),
which funded in part the work reported here.
This work was also supported in part by NSF Grant CCF-1619053,
and by Grant~2015816 from the United-States--Israel Binational
Science Foundation (BSF).
Portions of this work were conducted while
P.~H.\ Siegel visited Technion in May 2013 and
while R.~M.\ Roth visited the Center for
Memory and Recording Research (CMRR) at UC San Diego in
summer~2018.

\ifPAGELIMIT
\else
\appendices

\section{Alternative Definition of Capacity}
\label{sec:appendixA}

The definition of $\Capacity{q;a,b}{\bldp}$
in~\cite{CCK++18} differs from ours in~(\ref{def:Rqabp})
and takes the form
\begin{equation}
\label{def:Rqabp'}
\limsup_{n \to \infty} \frac{1}{n}
\log |\constraint_{q;a,b}(\bldp) \cap \Sigma^n|
\; ,
\end{equation}
where $\constraint_{q;a,b}(\bldp)$ consists of all
words $\bldw \in \constraint_{q;a,b}$ such that,
for some prescribed $s_0 \in \Low$,
the number of occurrences of any other symbol
$s \in \Sigma_0 := \Sigma \setminus \{ s_0 \}$
in~$\bldw$ equals $\left\lfloor p_s |\bldw| \right\rfloor$
(and $s_0$ fills up the remaining positions).

Clearly, for any fixed $\varepsilon > 0$ and sufficiently
large~$n$,
\[
\constraint_{q;a,b}(\bldp) \cap \Sigma^n \subseteq
\constraint_{q;a,b}(\bldp,\varepsilon) \cap \Sigma^n
\]
and, therefore, 
(\ref{def:Rqabp'}) is bounded from above by
$\Capacity{q;a,b}{\bldp}$ as defined in~(\ref{def:Rqabp}).

We next turn to showing that~(\ref{def:Rqabp})
is also a lower bound on~(\ref{def:Rqabp'}).
We assume here that $p_{\min} := \min_{s \in \Sigma_0} p_s > 0$;
the case where some entries of~$\bldp$ are zero can then be
argued by the continuity of
$\bldp \mapsto \Capacity{q;a,b}{\bldp}$
(at neighborhoods of vectors~$\bldp$ with $p_{\min} = 0$).\footnote{%
In particular, it can be verified that
the second case in Theorem~\ref{thm:R3p}---that of
$p_\Low = \frac{1}{2}$ and $p_\Int > 0$---indeed converges
to the first case therein.}

We define a one-to-one mapping from
$\constraint_{q;a,b}(\bldp,\varepsilon) \cap \Sigma^n$
into
$\constraint_{q;a,b}(\bldp) \cap \Sigma^{n'}$,
where
\[
n' = n + 1 + \left\lfloor (\varepsilon/p_{\min}) \cdot n \right\rfloor
\; ,
\]
namely, $n'$ is only ``slightly larger'' than~$n$.
The image of a word
$\bldw \in \constraint_{q;a,b}(\bldp,\varepsilon) \cap \Sigma^n$
is a word $\bldw' = \bldw \bldw''$,
where the suffix $\bldw''$ is determined as follows.
\begin{itemize}
\item
If~$\bldw$ ends with a symbol in $\Low \cup \Int$,
then~$\bldw''$ is any sequence of
symbols from~$\Low$, $\Int$, and~$\High$, in that order,
so that the count of each symbol $s \in\Sigma_0$
in $\bldw'$ reaches $\left\lfloor p_s n' \right\rfloor$.
The symbol $s_0$ then fills any vacant positions
among those allocated to symbols of~$\Low$.
\item
If~$\bldw$ ends with a symbol in $\High$, then we do the same except
that the symbols from~$\Low$ are filled last.
\end{itemize}

The count of each symbol $s \in \Sigma$ in~$\bldw$
is bounded from above by
\begin{equation}
\label{eq:CKK-1}
(p_s + \varepsilon)n \le
p_s(1 + \varepsilon/p_{\min}) n < p_s n' \; .
\end{equation}
In addition, from $\sum_{s \in \Sigma} p_s n' = n'$ we get
\[
\sum_{s \in \Sigma_0} \left\lfloor p_s n' \right\rfloor \le
n' - (p_{s_0} n')
\]
and, since we assume that $p_{s_0} > 0$,
\begin{equation}
\label{eq:CKK-2}
\sum_{s \in \Sigma_0} \left\lfloor p_s n' \right\rfloor < n' \; .
\end{equation}
It follows from~(\ref{eq:CKK-1})--(\ref{eq:CKK-2})
that we should always be able to reach
the targeted count, $\left\lfloor p_s n' \right\rfloor$,
in~$\bldw'$, for each symbol $s \in \Sigma_0$.
Moreover, by~(\ref{eq:CKK-2}),
there will be at least one position in~$\bldw''$ filled with~$s_0$;
so, when~$\bldw$ ends with a symbol in $\Low$,
that symbol will be followed in~$\bldw'$
by a symbol in~$\Low$.
Hence, the image~$\bldw'$ is indeed in
$\constraint_{q;a,b}(\bldp) \cap \Sigma^{n'}$.

We conclude that
\[
|\constraint_{q;a,b}(\bldp,\varepsilon) \cap \Sigma^n|
\le
|\constraint_{q;a,b}(\bldp) \cap \Sigma^{n'}| \; ,
\]
and taking logarithms, dividing by~$n'$,
and then taking $n \rightarrow \infty$ yields
that~(\ref{def:Rqabp'}) is at least
\[
\frac{1}{1+(\varepsilon/p_{\min})}
\limsup_{n \rightarrow \infty}
\frac{1}{n}
\log
|\constraint_{q;a,b}(\bldp,\varepsilon) \cap \Sigma^n|
\; .
\]
Finally, taking the limit $\varepsilon \rightarrow 0^+$
implies that our definition of
$\Capacity{q;a,b}{\bldp}$ in~(\ref{def:Rqabp}) is a lower bound
on~(\ref{def:Rqabp'}).

\section{Proof of $\partial z/\partial \lambda > 0$}
\label{sec:appendixB}

We show that if $z(y,\lambda)$ is as in~(\ref{eq:z}),
then $\partial z/\partial \lambda > 0$ whenever
$\lambda \ge y + 1 \ge 1$.

For $\lambda = y + 1$
the numerator in~(\ref{eq:z}) vanishes and therefore we have:
\[
\left.
\frac{\partial z}{\partial \lambda}
\right|_{\lambda= y+1} =
\left.
\frac{\lambda(3\lambda -2(y+1))}{\lambda^2 - \lambda + y+1}
\right|_{\lambda= y+1}  > 0 \; .
\]

When $\lambda > y + 1$ we have $z > 0$; we show that
in this case, $\partial \lambda(A(y,z))/\partial z > 0$.
Given $(y,z)$ where $y \ge 0$ and $z > 0$,
let $\bldx = \bldx(y,z) = (x_1 \; x_2 \; x_3)^T$ be
the unique (up to scaling)
all-positive eigenvector that corresponds to
$\lambda = \lambda(A(y,z))$:
\[
A(y,z) \bldx = \lambda \bldx \; .
\]
Let $\alpha$ be a positive real
in $(0,1)$ that satisfies the inequality:
\[
\alpha < x_1
\left( x_3 + \left( \frac{\lambda}{y+1} - 1 \right) x_2 \right)^{-1}
\; .
\]
It can be verified that for such an~$\alpha$,
the following inequality holds componentwise for
any sufficiently small $\delta > 0$:
\begin{equation}
\label{eq:PF}
A(y,z+\delta)
\left(
\begin{array}{c}
x_1 \\ x_2 \\ x_3 + \varepsilon
\end{array}
\right)
\ge
(\lambda + \alpha \cdot \delta)
\left(
\begin{array}{c}
x_1 \\ x_2 \\ x_3 + \varepsilon
\end{array}
\right)
\; ,
\end{equation}
where $\varepsilon = \alpha \delta x_2 /(y+1)$.
By Theorem~5.4 in\footnote{%
The proof of the necessity part in that theorem
does not require the entries of the matrix~$A$ or
the vector~$\bldx$ to be integers.}
\cite{MRS01} it follows from~(\ref{eq:PF})
that $\lambda(A(y,z+\delta)) \ge \lambda + \alpha \delta$.
This, in turn, implies that $\alpha$ is a lower bound on
$\partial \lambda(A(y,z))/\partial z$.

\section{Proof of Theorem~\ref{thm:R3p}}
\label{sec:appendixC}

Consider $\rho_\Int > 0$ first.
Requiring $\partial g / \partial y$ to be~$0$ yields
\begin{equation}
\label{eq:dgdy}
\rho_\Int (\lambda - y - 1)(\lambda^2 - \lambda + y + 1)
- \rho_\High \lambda^2 y = 0 \; ,
\end{equation}
and requiring $\partial g / \partial \lambda$ to be~$0$ yields
\begin{eqnarray}
\lefteqn{
(1-\rho_\High) (\lambda - y - 1) (\lambda^2 - \lambda + y + 1)
} \makebox[5ex]{} \nonumber \\
\label{eq:dgdlambda}
&&
{} - \rho_\High \left( \lambda^2 y + 2\lambda(y+1) - (y+1)^2 \right)
= 0 \; .
\end{eqnarray}
Rearranging terms in~(\ref{eq:dgdy}) in descending powers of~$y$
results in:
\begin{equation}
\label{eq:diff1}
\rho_\Int y^2
+ \Bigl(  \rho_\Int (\lambda^2 - 2 \lambda + 2)
+ \rho_\High \lambda^2 \Bigr) y
+ \rho_\Int (1-\lambda)(\lambda^2-\lambda+1) = 0 \; .
\end{equation}
Also, computing
$((1/\rho_\High) - 1) \times (\ref{eq:dgdy})
- (\rho_\Int/\rho_\High) \times (\ref{eq:dgdlambda})$,
and recalling that $\rho_\Low = 1 - \rho_\High - \rho_\Int$, we obtain
(after rearranging terms in descending powers of~$y$):
\begin{equation}
\label{eq:diff2}
\rho_\Int y^2
+ \Bigl( \rho_\Low \lambda^2 + 2 \rho_\Int(1-\lambda) \Bigr) y
+ \rho_\Int (1 - 2\lambda) = 0 \; .
\end{equation}
Finally, subtracting~(\ref{eq:diff1}) from~(\ref{eq:diff2})
results in:
\begin{equation}
\label{eq:diff3}
(1-2\rho_\Low) y - \rho_\Int(\lambda-2) = 0 \; .
\end{equation}

The case $\rho_\Low = \half$ is somewhat special,
as~$y$ then disappears from~(\ref{eq:diff3}) and we get
$\lambda = 2$. Plugging this value of~$\lambda$ into~(\ref{eq:diff1})
(or into~(\ref{eq:diff2})) yields the following equation for~$y$:
\begin{equation}
\label{eq:pL=1/2}
\rho_\Int y^2 + 2(\rho_\Int + 2\rho_\High)y - 3 \rho_\Int = 0 \; .
\end{equation}
We solve the quadratic equation~(\ref{eq:pL=1/2}) for the positive root:
\begin{equation}
\label{eq:ypL=1/2}
y = - 1 - 2\tau + 2 \sqrt{1 + \tau + \tau^2} \; ,
\end{equation}
where $\tau = \rho_\High/\rho_\Int$.
It is easy to check that $(y,2) \in \Domain$.
This yields the case of $\rho_\Low = \half$ and $\rho_\Int > 0$
in the statement of the theorem.

Assume hereafter that $\rho_\Low \ne \half$. From~(\ref{eq:diff3})
we can express~$y$ in terms of~$\lambda$:
\begin{equation}
\label{eq:ypLne1/2}
y = \frac{\rho_\Int(\lambda-2)}{1-2\rho_\Low} \; .
\end{equation}
Substituting this value into~(\ref{eq:diff2}) yields:
\begin{eqnarray*}
\lefteqn{
\rho_\Int^2 (\lambda-2)^2
+ (1-2\rho_\Low)(\lambda-2)
\left( \rho_\Low \lambda^2 + 2\rho_\Int(1-\lambda) \right)
} \makebox[12ex]{} \\
&&{} + (1-2\rho_\Low)^2(1-2\lambda) = 0 \; .
\end{eqnarray*}
This is a cubic equation in $\lambda$:
$Z(\lambda) = 0$,
where the coefficients of $Z(x) = \sum_{i=0}^3 x^i$ are given by
\begin{eqnarray*}
Z_3 &=& (1-2\rho_\Low) \rho_\Low \\
Z_2 &=& \rho_\Int^2 - 2(1-2\rho_\Low)(1-\rho_\High) \\
Z_1 &=& -4 \rho_\Int^2 + 2(1-2\rho_\Low)(3 \rho_\Int + 2\rho_\Low - 1)\\
Z_0 &=& 4 \rho_\Int^2 - 4(1-2\rho_\Low) \rho_\Int + (1-2\rho_\Low)^2\; .
\end{eqnarray*}
Plugging $\rho_\Int = 1 - \rho_\Low - \rho_\High$ into
these expressions yields:
\begin{eqnarray}
Z(x) & \!\!\!\!=\!\!\!\! & (1-2\rho_\Low) \rho_\Low x^3
+ \left( (\rho_\Low-\rho_\High)^2 - (1-2\rho_\Low) \right) x^2 
\nonumber \\
\label{eq:Zx}
&& {} - 2(\rho_\Low -\rho_\High)(1-2\rho_\High) x + (1 - 2\rho_\High)^2 \; .
\end{eqnarray}

Next, we find a root $\lambda$ of $Z(x)$ such that,
along with~$y$ as in~(\ref{eq:ypLne1/2}),
we get a point in $\Domain$. This point necessarily attains
the infimum in~(\ref{eq:R3p:3}).
For the analysis, we will find it useful
to re-write $Z(x)$ as
\begin{equation}
\label{eq:Z}
Z(x) = (1-2\rho_\Low)(\rho_\Low x - 1) x^2
+ \Bigl( (\rho_\Low-\rho_\High) x - (1-2\rho_\High) \Bigr)^2 \; .
\end{equation}

We distinguish between two cases.

\emph{Case 1: $\rho_\Low < \half$}.
It is easily seen that $Z(0) = (1 - 2\rho_\High)^2 \ge 0$;
moreover, if $Z(0) = 0$ then~$0$ is a multiple root of $Z(x)$.
From~(\ref{eq:Z}), it readily follows that
\[
Z(2) = -3(1-2\rho_\Low)^2 < 0 \; .
\]
We conclude that $Z(x)$
has three real roots: one in $(2,\infty)$, a second root in $[0,2)$,
and a third root which is non-positive. Since~$y$ in~(\ref{eq:ypLne1/2})
has to be positive, it follows that~$\lambda$ equals
the unique root of $Z(x)$ which is in $(2,\infty)$.

\emph{Case 2: $\rho_\Low > \half$}.
   From~(\ref{eq:Z}), we also get that
$Z(1/\rho_\Low) =
\left( (\rho_\Low-\rho_\High)/\rho_\Low - (1-2\rho_\High) \right)^2 =
\rho_\High^2 \left(2 - (1/\rho_\Low) \right)^2 > 0$.
Moreover, it is easily seen from~(\ref{eq:Z})
that $Z(x)$ is positive on the entire interval $(-\infty,1/\rho_\Low]$.
Since $Z(2) < 0$, it follows that $Z(x)$ has a root in
$(1/\rho_\Low,2)$ (and that root is the smallest real root of $Z(x)$).
For such a root the value~$y$ in~(\ref{eq:ypLne1/2})
is positive.\footnote{%
In fact, $Z(x)$ must have a \emph{unique} root in $(1/\rho_\Low,2)$:
otherwise, there would be two points,
$(y^{(1)},z^{(2)})$ and
$(y^{(2)},z^{(2)})$,
with $z^{(1)} \ne z^{(2)}$, that would attain the infimum
in~(\ref{eq:R3p:2}).
These points would correspond in a one-to-one manner to two distinct points,
$(\xi_\Int^{(1)},\xi_\High^{(2)})$ and
$(\xi_\Int^{(2)},\xi_\High^{(2)})$, that would attain
the infimum of~(\ref{eq:R3p:1}).
Yet then, by the convexity of~(\ref{eq:R3p:1}),
all the points on the line that connects the latter two points
would attain
the infimum, thereby absurdly implying that $Z(x)$ has infinitely
many roots.}

In both cases, we have
\[
y = \frac{\rho_\Int(\lambda-2)}{1-2\rho_\Low} \le
\frac{(1-\rho_\Low)(\lambda-2)}{1-2\rho_\Low} < \lambda - 1
\; ,
\]
i.e., $(y,\lambda) \in \Domain$. This completes the analysis of
the $\rho_\Int > 0$ case.

To deal with the $\rho_\Int = 0$ case, we set $y = 0$.
Recall that $g_0(\lambda) := g(0,\lambda)$.
Setting $g_0'(\lambda) = 0$ yields~(\ref{eq:dgdlambda})
with $y = 0$. If $\rho_\Low = \half$, then $\rho_\High = \half$ as well,
and~(\ref{eq:dgdlambda}) simplifies
to $\lambda^3 - 2\lambda^2 = 0$, from which we obtain
$\lambda = 2$ as the only solution larger than $1$.
This yields the case of $\rho_\Low = \half$ and $\rho_\Int = 0$ in
the theorem statement.

When $\rho_\Low \ne \half$,
then using the fact that $1-\rho_\High = \rho_\Low$
(since $\rho_\Int = 0$), we write the LHS
of~(\ref{eq:dgdlambda}) (again with $y = 0$) as
\[
Z_0(\lambda) :=
\rho_\Low \lambda^3 - 2\rho_\Low \lambda^2 + 2(2\rho_\Low-1)
\lambda - (2\rho_\Low-1) \; .
\]
On the other hand, using the fact that
$1-2\rho_\Low = \rho_\High-\rho_\Low = 2\rho_\High - 1$,
it can be verified that
$Z(x)$ as in~(\ref{eq:Zx}) is in fact equal to $(1-2\rho_\Low) Z_0(x)$.
Since $\rho_\Low \ne \half$, $Z(x)$ and $Z_0(x)$ have
the same roots, from which it follows that the theorem statements for
the $\rho_\Low \ne \half$ case apply when $\rho_\Int = 0$ as well.
\fi
\end{document}